\documentclass[12pt,draftclsnofoot,journal,onecolumn]{IEEEtran}

\usepackage{graphicx}
\usepackage[cmex10]{amsmath}
\usepackage{amsfonts}
\usepackage{amssymb}
\usepackage{mathrsfs}
\usepackage{bm}
\usepackage{algorithmic}
\usepackage{algorithm}
\usepackage{array}
\usepackage{mdwmath}
\usepackage{mdwtab}
\usepackage[caption=false,font=footnotesize]{subfig}
\usepackage{fixltx2e}
\usepackage{stfloats}
\usepackage{footnote}
\usepackage{tabularx}
\usepackage{multirow}
\usepackage{color}
\usepackage{textcomp}
\usepackage{subfig}

\newtheorem{theorem}{Theorem}
\newtheorem{lemma}{Lemma}

\begin{document}

\title{Optimal Foresighted Multi-User Wireless Video}

\author{Yuanzhang~Xiao and~Mihaela~van~der~Schaar,~\IEEEmembership{Fellow,~IEEE} \\
Department of Electrical Engineering, UCLA. \{yxiao,mihaela\}@ee.ucla.edu.}

\maketitle

\begin{abstract}
Recent years have seen an explosion in wireless video communication systems. Optimization in such systems is crucial -- but most existing methods
intended to optimize the performance of multi-user wireless video transmission are inefficient. Some works (e.g. Network Utility Maximization (NUM))
are \emph{myopic}: they choose actions to maximize \emph{instantaneous} video quality while ignoring the future impact of these actions. Such myopic
solutions are known to be inferior to \emph{foresighted} solutions that optimize the long-term video quality. Alternatively, foresighted solutions
such as rate-distortion optimized packet scheduling focus on single-user wireless video transmission, while ignoring the resource allocation among
the users.

In this paper, we propose an optimal solution for performing joint \emph{foresighted} resource allocation and packet scheduling among \emph{multiple}
users transmitting video over a shared wireless network. A key challenge in developing foresighted solutions for multiple video users is that the
users' decisions are coupled. To decouple the users' decisions, we adopt a novel dual decomposition approach, which differs from the conventional
optimization solutions such as NUM, and determines foresighted policies. Specifically, we propose an informationally-decentralized algorithm in which
the network manager updates resource ``prices'' (i.e. the dual variables associated with the resource constraints), and the users make individual
video packet scheduling decisions based on these prices. Because a priori knowledge of the system dynamics is almost never available at run-time, the
proposed solution can learn online, concurrently with performing the foresighted optimization. Simulation results show 7 dB and 3 dB improvements in
Peak Signal-to-Noise Ratio (PSNR) over myopic solutions and existing foresighted solutions, respectively.
\end{abstract}

\vspace{-0.1cm}

\section{Introduction}
Video applications, such as video streaming, video conferencing, remote teaching, surveillance etc., have become the major applications deployed over
the current cellular networks and wireless Local Area Networks (LANs). The bandwidth-intensive and delay-sensitive video applications require
efficient allocation of network resources (e.g. bandwidth in 4G LTE networks or temporal transmission opportunity in wireless LANs) among the users
accessing the network, and efficient scheduling of each user's video packets based on its allocated resources.


Most existing solutions for multi-user wireless video transmission are \emph{myopic} \cite{vanderSchaarAndreopoulosHu}--\cite{JiHuangChiang}, meaning
that these joint resource allocation and packet scheduling solutions are designed to maximize the \emph{instantaneous}\footnote{We will define
instantaneous and long-term video quality rigorously in Section~\ref{sec:model}.} video quality (i.e. the average distortion impact of the packets
sent within the next transmission opportunity or time interval). These solutions can be interpreted as maximizing the instantaneous video quality by
repeatedly and independently solving over time Network Utility Maximization (NUM) problems. However, current (resource allocation and packet
scheduling) decisions impact the future system performance, which is not taken into consideration by the repeated NUM solutions. Optimal solutions
for multi-user resource allocation and packet scheduling need to be formalized as sequential decisions given (unknown) dynamics (i.e. time-varying
channel conditions and video traffic characteristics, dependencies across video packets etc.). Hence, the repeated NUM solutions are myopic and
inferior to \emph{foresighted} solutions that maximize the \emph{long-term} video quality.

To address this limitation, several foresighted solutions have been proposed for packet scheduling, see e.g. \cite{ChouMiao}--\cite{FuVDS_CSVT2012}.
However, these packet scheduling solutions focus on the sequential decision making of a single foresighted video user and do not consider the
coupling among users. In multi-user wireless networks, the solutions developed for a single user have been shown to be highly inefficient
\cite{FuVDS_JSAC2010}\footnote{The work \cite{FuVDS_JSAC2010} is the only one that develops foresighted solutions for multiple video users. We will
discuss the differences between our work and \cite{FuVDS_JSAC2010} in detail in Section~\ref{sec:related}.}, since they ignore that users are sharing
the same wireless resource and hence, their decisions over time are coupled. A simple solution to perform multi-user resource allocation was proposed
in \cite{vanderSchaarAndreopoulosHu}. However, this allocation is performed statically, possibly before run-time, and as shown in
\cite{FuVDS_JSAC2010}, such static resource allocation is suboptimal compared to solutions that dynamically allocate resources among multiple users
given the users' video traffic and channel characteristics which are time-varying and often unknown before run-time.

In this paper, we propose an optimal solution for performing joint \emph{foresighted} resource allocation and packet scheduling among multiple users
transmitting video over a wireless network. We consider the transmission of video by multiple users over a 4G cellular uplink network. In such
networks, the base station (BS) needs to decide how to allocate wireless resources (i.e. bandwidth) among multiple video users, each performing video
packet scheduling over the wireless network given the allocated resources. In our proposed solution, the BS does not directly allocate the resources;
instead, it mediates the resource allocation by charging each user a unit resource ``price''\footnote{Note that the ``price'' is a control signal,
rather than the price for real monetary payment.}. Given the price, each user determines its own optimal packet scheduling and resource acquisition.
Hence, our approach is decentralized and enables users to make optimal decisions locally, in an informationally-decentralized manner, based on their
own private information and the resource price. We propose a low-complexity algorithm in which the BS updates the resource prices using a stochastic
subgradient method based on the users' resource usage while the users make foresighted decisions based on these prices. We prove that the algorithm
can converge to the optimal prices, under which the users' optimal decisions maximize the long-term network utility. Moreover, our solution also
allows to impose a minimum video quality guarantee for each user (i.e. the video quality for each user needs to be higher than a preset minimum video
quality guarantee). Importantly, the BS and the users may not have in practice the statistic knowledge of the system dynamics (e.g. the incoming
video traffic characteristics and the time-varying channel conditions). For such scenarios, we propose a post-decision state (PDS) based learning
algorithm that converges to the proposed optimal solution. Finally, a desirable feature of the proposed solution is that the efficiency of the
proposed resource allocation scheme does not heavily depend on the adopted packet scheduling algorithm. In other words, the proposed resource
allocation can be deployed in conjunction with other packet scheduling algorithms (such as earliest-deadline-first scheduling) that are simpler but
suboptimal compared to the proposed scheduling algorithm. As we will show in our simulations section, our solution outperforms existing resource
allocation schemes even when users deploy such simple packet scheduling algorithms. Hence, it can be easily deployed by the network managers (e.g.
base stations) in conjunction with existing users' scheduling technology and applications. This allows for easy and gradual adoption of our proposed
methods.

Our proposed framework can be extended for deployment in other scenarios. First, although we focus on uplink video transmission, our work can be
easily extended to downlink video transmission, to IEEE 802.11a/e wireless LANs (Local Area Networks) in which transmission opportunities are
allocated among users \cite{vanderSchaarAndreopoulosHu}, and to IEEE 802.15.4e-enabled Internet-of-Things and machine-to-machine communications.
Second, although in this paper we focus on the users' decisions on how to schedule their packets given a pre-encoded video bitstream, the framework
can be extended to the case where the users' decisions are how to make real-time encoding decisions such as determining the Group of Pictures
structure and quantization parameters in response to resource prices.

The rest of the paper is organized as follows. We discuss prior work in Section~\ref{sec:related}. We describe the system model in
Section~\ref{sec:model} and formulate the design problem in Section~\ref{sec:Problem}. Then we propose our solution in
Section~\ref{sec:DesignFramework}. Simulation results in Section~\ref{sec:Simulation} demonstrate the performance improvement of the proposed
solution. Finally, Section~\ref{sec:Conclusion} concludes the paper.

\section{Related Works}\label{sec:related}

\subsection{Related Works on Video Transmission}
\begin{table}
\renewcommand{\arraystretch}{1.0}
\caption{Comparisons With Related Works.} \label{table:RelatedWork} \centering
\begin{tabular}{|c|c|c|c|c|c|}
\hline
 & Traffic model & Users & Foresighted & Learning & Optimal \\
\hline
\cite{ReibmanBerger}--\cite{YuHsiuPang_TMC2012} & Flow-level & Multiple & No & No & No \\
\hline
\cite{vanderSchaarAndreopoulosHu}--\cite{JiHuangChiang} & Packet-level & Multiple & No & No & No \\
\hline
\cite{ChouMiao}--\cite{FuVDS_CSVT2012} & Packet-level & Single & Yes & Yes & No \\
\hline
\cite{FuVDS_JSAC2010} & Packet-level & Multiple & Yes & Yes & No \\
\hline
Proposed & Packet-level & Multiple & Yes & Yes & Yes \\
\hline
\end{tabular}
\end{table}

The existing works on wireless video communications can be classified based on various criteria. First, some works
\cite{ReibmanBerger}--\cite{YuHsiuPang_TMC2012} model the video traffic with a simplified model, e.g. flow-level models using priority queues. Such a
model cannot accurately capture the heterogeneous distortion impact, delay deadlines, and dependency of the video packets. Hence, the solution
derived based on simplified models may not perform well in practice \cite{FuVDS_JSAC2010}\cite{FuVDS_CSVT2012}.

A plethora of recent works \cite{vanderSchaarAndreopoulosHu}--\cite{FuVDS_CSVT2012} adopt packet-level models for video traffic but they adopt
distinct solutions to optimize the video quality. Some works \cite{vanderSchaarAndreopoulosHu}--\cite{JiHuangChiang} assume that the users are
\emph{myopic}, namely they only maximize their \emph{instantaneous} video quality over a given time interval without considering the impact of their
actions on the \emph{long-term} video quality. They cast the problem in a NUM framework to maximize the instantaneous joint video quality of all the
users, and apply the NUM framework repeatedly when the channel conditions or video traffic characteristics change. However, since the users are
optimizing their transmission decisions myopically, their long-term average performance is inferior to the performance achieved when the users are
\emph{foresighted} \cite{FuVDS_CSVT2012}. Most of the works \cite{ChouMiao}--\cite{FuVDS_CSVT2012} considering the foresighted decision of users
focus solely on a \emph{single} foresighted user making sequential transmission decisions (e.g. packet scheduling, retransmissions etc.). Among them,
some \cite{ChouMiao}--\cite{WangOrtega} consider the dynamics of video traffic only, while others \cite{FuVDS_CSVT2012} consider the dynamics of both
video traffic and channel conditions. However, these single user solutions do not discuss how to allocate resources among multiple users as well as
how this allocation is impacted by and impacts the foresighted scheduling decisions of individual users. Static allocations of resources which are
often assumed in the works studying the foresighted decisions of single users for wireless transmission have been shown to be suboptimal compared to
the solutions that dynamically allocate resources among multiple users \cite{FuVDS_JSAC2010}.

The only work which proposes a solution for video resource allocation among multiple foresighted users is \cite{FuVDS_JSAC2010} . Since this work
\cite{FuVDS_JSAC2010} is most related to our proposed solution, we discuss the differences between them in detail. The challenge in foresighted
multi-user video transmission is that the users' decisions are dynamic and coupled through the resource (e.g. bandwidth or time) constraints. Hence,
the design problem is much more complicated than in myopic multi-user video resource allocation and transmission (e.g. NUM) or in foresighted
single-user video transmission given a static or pre-determined resource allocation. To overcome this challenge, a dual decomposition method has been
proposed in \cite{FuVDS_JSAC2010}, which removes the resource constraints in all the states and adds them to the objective functions in the
corresponding states as penalties modulated by the \emph{same} Lagrangian multiplier (interpreted as the price of resources). There are two key
differences between the dual decomposition method in \cite{FuVDS_JSAC2010} and the classical dual decomposition method used in myopic optimization
such as in NUM. The first difference is that the users' objective functions are the long-term video quality, and hence the decomposed subproblems
solved by each user are foresighted optimization problems instead of static optimization problems as in NUM. The second one is that the update of the
dual variables is different from that in NUM since the subproblems are foresighted. However, such a solution is in general suboptimal (i.e. there is
a positive duality gap) for the foresighted multi-user optimization for the following reason. Since it uses the same Lagrangian multiplier (i.e. a
uniform price) for the resource constraints in all the states \cite{FuVDS_JSAC2010}, the uniform price is usually set high such that the resource
constraints in all the states are satisfied. Technically, when there are some active constraints (i.e. constraints satisfied with strict inequality),
the complementary slackness conditions cannot be satisfied under all the states. Hence, there will always be a duality gap due to the violation of
complementary slackness conditions. In contrast, in our work, we use different prices for resource constraints in different states (i.e. different
channel conditions). In this way, we can achieve the optimal performance (i.e. there is no duality gap, unlike the MU-MDP solution with uniform price
in \cite{FuVDS_JSAC2010}).

Table~\ref{table:RelatedWork} summarizes the above discussions. Note that the optimality shown in the last column of Table~\ref{table:RelatedWork}
indicates whether the solution is optimal for the long-term network utility (i.e. the joint long-term video quality of all the users in the network).
The existing solutions can be optimal in the corresponding models (e.g. the solutions based on repeated NUM
\cite{vanderSchaarAndreopoulosHu}--\cite{JiHuangChiang} are optimal when the users are myopic).


\begin{table}
\renewcommand{\arraystretch}{1.0}
\caption{Comparisons With Related Theoretical Frameworks.} \label{table:RelatedWork_MathematicalFramework} \centering
\begin{tabular}{|c|c|c|c|c|c|}
\hline
\multirow{2}{*}{} & Decision & Decentralized & \multirow{2}{*}{Foresighted} & \multirow{2}{*}{Learning} & \multirow{2}{*}{Optimal} \\
                  & makers   & information   &                              &                           & \\
\hline
MDP \cite{FuVDS_CSVT2012} & Single & N/A & Yes & Yes & No \\
\hline
Repeated NUM & \multirow{2}{*}{Multiple} & \multirow{2}{*}{Yes} & \multirow{2}{*}{No} & \multirow{2}{*}{No} & \multirow{2}{*}{No} \\
\cite{JiHuangChiang} & & & & & \\
\hline
MU-MDP \cite{FuVDS_JSAC2010} & Multiple & Yes & Yes & Yes & No \\
\hline
Lyapunov & & & & & \\
Optimization & \multirow{2}{*}{Single} & \multirow{2}{*}{N/A} & \multirow{2}{*}{Yes} & \multirow{2}{*}{Yes} & \multirow{2}{*}{No} \\
\cite{Neely} & & & & & \\
\hline
Proposed & Multiple & Yes & Yes & Yes & Yes \\
\hline
\end{tabular}
\end{table}

\subsection{Related Theoretical Frameworks}
Single-user foresighted decision making in a dynamically changing environment has been studied and formulated as Markov Decision Process (MDP). As
mentioned previously, the problem of multiple video users making coupled foresighted decisions (MU-MDPs) has been studied in \cite{FuVDS_JSAC2010}.
However, as we discussed before, this MU-MDP solution is based on uniform prices and is therefore suboptimal for most multi-user wireless video
scenarios, even though it has been shown to significantly outperform the myopic solutions.

It is worth noting that foresighted decision making in a dynamically changing environment can also be solved using a Lyapunov optimization framework
\cite{Neely}. However, the Lyapunov optimization framework is not able to take optimal decisions for wireless video streaming since it disregards the
specific delay-constraints and distortion-impacts of the video traffic \cite{FuvanderSchaar_TVT2012}.

Table~\ref{table:RelatedWork_MathematicalFramework} summarizes the above discussions about existing theoretical frameworks.

In Table~\ref{table:DetailedComparison_Frameworks} at the end of Section~\ref{sec:DesignFramework}, we provide rigorous technical comparisons with
existing works after we describe our proposed solution in detail.


\section{System Model}\label{sec:model}
We first present a general abstract model for multi-user wireless video transmission in 4G LTE cellular networks or IEEE 802.11a/e wireless LANs.
Then we give an example of a packet-level video transmission model as in \cite{ChouMiao}\cite{FuVDS_CSVT2012}\cite{FuVDS_JSAC2010} to illustrate our
abstract model.

\subsection{The General Model}
We consider a network with a network manager (e.g. the base station), indexed by $0$, and a set $\mathcal{I}$ of $I$ wireless video users, indexed by
$i=1,\ldots,I$. Time is slotted at $t=0,1,2,\ldots$. In the rest of the paper, we will put the user index in the superscript and the time index in
the subscript of variables. The multi-user wireless video transmission system is described by the following features:
\begin{itemize}
\item \emph{States}: Each user $i$ has a state $s^i \in S^i$, which is realized and known to user $i$ at the beginning of each time slot. The state $s^i$ consists of two parts, namely the
video traffic state $T^i$ and the channel state $h^i$. An example of a simplified (video) traffic state can be the types of video frames (I, P, or B
frame) available for transmission and the numbers of packets in each available video frame. An example of a channel state can be the channel quality
reported to the application layer by the lower layers. The network manager has a state $s^0 \in S^0$, which consists of all the users' channel
states, namely $s^0=(h^1,\ldots,h^I)$. In uplink transmission, the channel states of the users can be estimated by the network manager and fed back
to the users. Similarly, in the downlink, each user estimates its own channel state and feeds it back to the network manager.
\item \emph{Actions}: At each state $s^i$, each user $i$ chooses a packet scheduling action $a^i \in A^i(s^i)$. A packet scheduling action determines
how many packets of each available video frame should be transmitted. We allow the sets of actions taken under different states to be different, in
order to incorporate the minimum video quality requirements that we will discuss later.
\item \emph{Payoffs}: Each user $i$ has a payoff function $u^i: S^i \times A^i \rightarrow \mathbb{R}$. A typical payoff can be the distortion impact
of the transmitted packets plus the (negative) disutility incurred by the energy consumption in transmitting the packets.
\item \emph{State Transition}: Each user $i$'s next state depends on its current state and its current action. We denote the state transition by
$p^i(s^{i\prime}|s^i,a^i) \in \Delta(S^i)$, where $\Delta(S^i)$ is the probability distribution over the set of states.
\item \emph{Resource Constraints}: Given the users' channel states $(h^1,\ldots,h^I)$, the users' actions need to satisfy a resource constraint (e.g. the total
bandwidth or transmission opportunity is limited). Note that the resource constraint does not depend on the users' traffic states. Since the network
manager's state is the collection of the users' channel states, we write the resource constraint compactly as
$$
f(s^0,a^1,\ldots,a^I) \leq 0.
$$
\end{itemize}

A variety of multi-user wireless video transmission systems can be abstracted using the above general model. Next, we present a packet-level video
transmission model as an illustration.

\subsection{An Example Packet-Level Video Transmission Model}
Packet-level video transmission models have been proposed in a variety of related works, including \cite{ChouMiao}--\cite{FuVDS_JSAC2010}. In the
following, we briefly describe the model as an illustration of our general model, and refer interested readers to
\cite{ChouMiao}--\cite{FuVDS_JSAC2010} for more details.

We first consider a specific video user $i$, and hence drop the superscript before we discuss the resource constraints. The video source data is
encoded using an H.264 or MPEG video coder under a Group of Pictures (GOP) structure: the data is encoded into a series of GOPs, indexed by
$g=1,2,\ldots$, where one GOP consists of $N$ data units (DUs). Each DU $n=1,\ldots,N$ in GOP $g$, denoted $DU_n^g$, is characterized by its size
$l_n^g \in \mathbb{N}_+$ (i.e. the number of packets in it), distortion impact $q_n^g \in \mathbb{R}_+$, delay deadline $d_n^g \in \mathbb{N}_+$, and
dependency. The dependency of the DUs in a GOP is represented by a directed acyclic graph (DAG) \cite{ChouMiao}, where the vertices are DUs and an
edge from $DU_m^g$ to $DU_n^g$ indicates that $DU_n^g$ depends on $DU_m^g$. The dependency among the DUs in one GOP comes from encoding techniques
such as motion estimation/compensation. In general, if $DU_n^g$ depends on $DU_m^g$, we have $d_n^g \geq d_m^g$ and $q_n^g \leq q_m^g$, namely
$DU_m^g$ should be decoded before $DU_n^g$ and has a higher distortion impact than $DU_n^g$ \cite{FuVDS_CSVT2012}.

Among the above characteristics, the distortion impact $q_n^g$, delay deadline $d_n^g$, and the dependency are deterministic and fixed for the same
DUs across different GOPs (e.g. $q_n^g=q_n^{g+1}$) \cite{FuVDS_CSVT2012}\cite{FuVDS_JSAC2010}. The size of $DU_n^g$ is random following the
probability mass function $PMF_n$. As in \cite{FuVDS_JSAC2010}, we assume that the sizes of different DUs are independent random variables. Note that
the distributions of the sizes of the $n$th DUs in different GOPs are the same.

\begin{figure}
\centering
\includegraphics[width =3.5in]{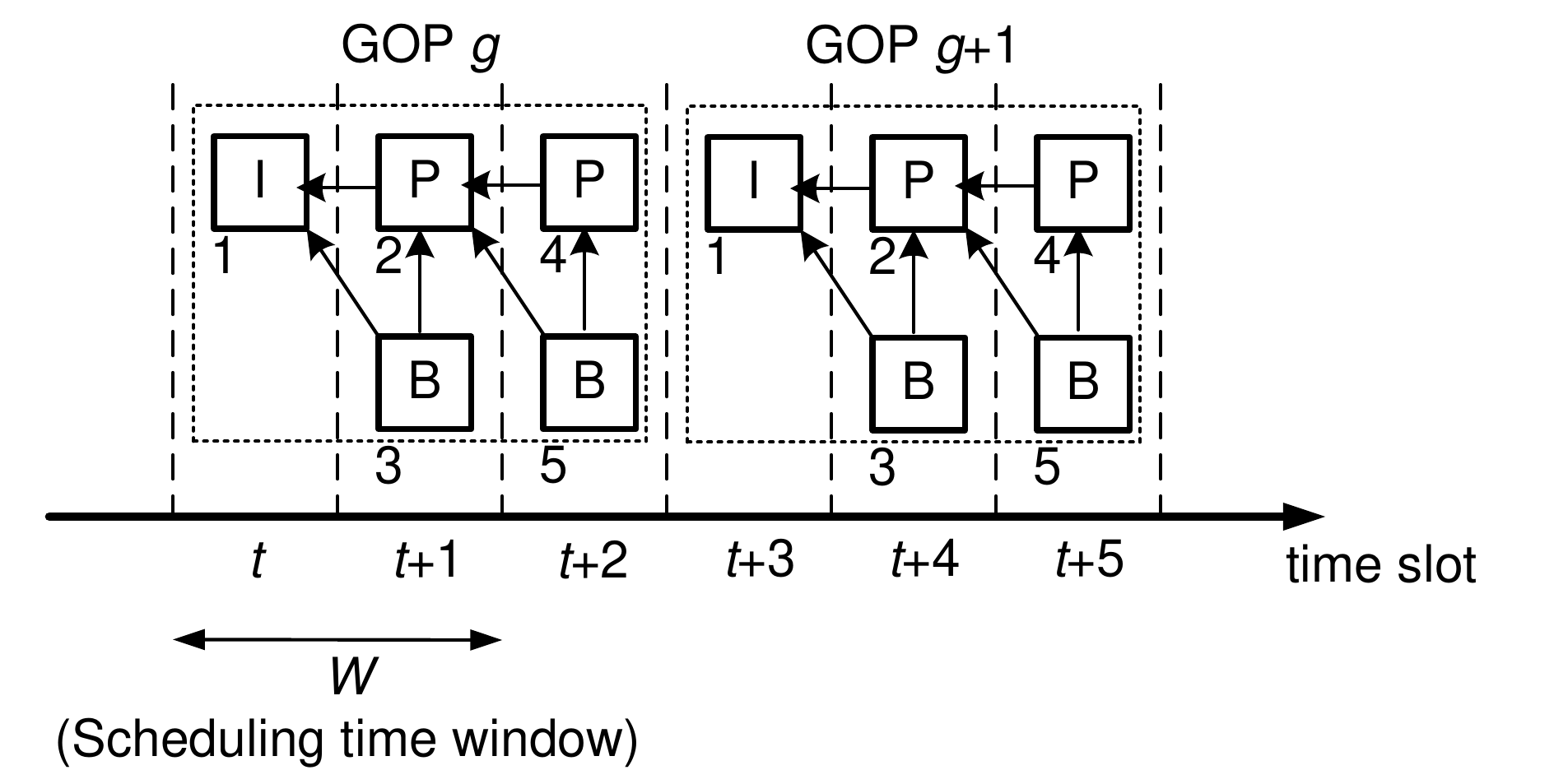}
\caption{Illustration of GOP (group of pictures), DU (data unit), and packet scheduling. Since the scheduling time window is $W=2$, the contexts in
different time slots are $C_t=\{DU_1^g,DU_2^g,DU_3^g\}$, $C_{t+1}=\{DU_2^g,DU_3^g,DU_4^g,DU_5^g\}$, $C_{t+2}=\{DU_4^g,DU_5^g,DU_1^{g+1}\}$,
$C_{t+3}=\{DU_1^{g+1},DU_2^{g+1},DU_3^{g+1}\}$, and so on.} \label{fig:SystemModel}
\end{figure}

\subsubsection{The Traffic State} Next we discuss how to construct the traffic state $T_t$. At time slot $t$, as in \cite{ChouMiao}\cite{FuVDS_CSVT2012}\cite{FuVDS_JSAC2010}, we assume that
the wireless user will only consider for transmission the DUs with delay deadlines in the range of $[t, t + W)$, where $W$ is referred to as the
scheduling time window (STW). We further assume as in \cite{ChouMiao}\cite{FuVDS_CSVT2012}\cite{FuVDS_JSAC2010} that the STW is chosen such that any
two DUs that have direct dependency can be considered for transmission in the same time slot. Following the model in
\cite{FuVDS_CSVT2012}\cite{FuVDS_JSAC2010}, at time slot $t$, we introduce a context to represent the set of DUs that are considered for
transmission, i.e., whose delay deadlines are within the range of $[t, t + W)$. We denote the context by $C_t = \left\{DU^g_j|d^g_j \in [t, t +
W)\right\}$. Since the GOP structure is fixed, the context $C_t$ is periodic with the duration $T$ of a GOP, which means that for any $DU^g_j \in
C_t$ , there exists $DU^{g+1}_j \in C_{t+T}$. Hence, $C_t$ and $C_{t+T}$ have the same types of DUs and the same DAG between these DUs. The
transition from context $C_t$ to $C_{t+1}$ is deterministic. An illustration of the context is given in Fig.~\ref{fig:SystemModel}.

Given the current context $C_t$, we let $x_{t,DU}$ denote the number of packets in the buffer associated with $DU \in C_t$. We denote the buffer
state of the DUs in $C_t$ by $x_t = \left\{x_{t,DU} | DU \in C_t\right\}$. The traffic state $T_t$ at time slot $t$ is then defined as $T_t = (C_t,
x_t)$, where the context represents the types of DUs, the dependencies among them, and the buffer state $x_t$ represents the amount of packets
remaining for transmission.

\subsubsection{The Channel State} At each time slot $t$, the wireless user experiences a channel condition $h_t \in \mathcal{H}$, where $\mathcal{H}$ is the set of finite
possible channel conditions and $h_t$ is referred to as the channel state. In this paper, we assume that the wireless channel is slow-fading (e.g.
remains the same in one time slot) and the channel condition $h_t$ can be modeled as a finite-state Markov chain with transition probability
$p_h(h^\prime|h) \in [0, 1]$ \cite{ZhangKassam}. Hence, the state which the wireless user experiences at each time slot $t$ is $s_t = (C_t, x_t,
h_t)$, which includes the current context, buffer state and channel state.

\subsubsection{Packet Scheduling Action} At time slot $t$, the wireless user decides how many packets should be transmitted from each DU in the current context. The decision is
represented by $a_t(C_t, x_t, h_t) = \{y_{t,DU} | DU \in C_t, y_{t,DU} \in [0, x_{t,DU}]\}$, where $y_{t,DU}$ represents the amount of packets
transmitted from DU $f$.

%

\subsubsection{Payoff} As in \cite{FuVDS_CSVT2012}, we consider the following \emph{instantaneous payoff} at each time slot $t$:
\begin{eqnarray}
u(s_t,a_t) = \sum_{DU \in C_t} q_{DU} y_{t,DU} - \beta \cdot \rho\left(h_t, \|a_t\|_1\right),
\end{eqnarray}
where the first term $\sum_{DU \in C_t} q_{DU} y_{t,DU}$ is the \emph{instantaneous video quality}, namely the distortion reduction obtained by
transmitting the packets from the DUs in the current context, and the second term $\beta \cdot \rho\left(h_t, \|a_t\|_1\right)$ represents the
disutility of the energy consumption by transmitting the packets. Since the packet scheduling action $a_t$ is a vector with nonnegative components,
we have
$$
\|a_t\|_1=\sum_{DU \in C_t} q_{DU} y_{t,DU},
$$
namely $\|a_t\|_1$ is the total number of transmitted packets. As in \cite{FuVDS_CSVT2012}, the energy consumption function $\rho(h, \|a\|_1)$ is
assumed to be convex in the total number of transmitted packets $\|a\|_1$ given the channel condition $h$. An example of such a function can be
$\rho(h,\|a\|_1)=\sigma^2(e^{2\|a\|_1 b}-1)/h$, where $b$ is the number of bits in one packet \cite{BertsekasGallager}. The payoff function is a
tradeoff between the distortion reduction and the energy consumption, where the relative importance of energy consumption compared to distortion
reduction is characterized by the tradeoff parameter $\beta>0$. In the simulation, we will set different values for $\beta$ to illustrate the
tradeoff between the distortion reduction and energy consumption.

\subsubsection{The Resource Constraint} In this work, we consider a 4G cellular network,
where there are $I$ wireless video users transmitting to the BS indexed by $0$. The users access the channels in a FDMA (frequency-division multiple
access) manner. We normalize the total bandwidth to be $1$, and will be divided and shared by the users.

We assume that each user $i$ uses adaptive modulation and coding (AMC) based on its channel condition. In other words, each user $i$ chooses a data
rate $r_t^i$ under the channel state $h_t^i$. Note that the rate selection is done by the physical layer and is not a decision variable in our
framework. Then as in \cite{LiZhangZhaoRangarajan_ICNP2010}\cite{YuHsiuPang_TMC2012}, we have the following resource constraint:
\begin{eqnarray}\label{eqn:BandwidthConstraint}
\sum_{i=1}^I \frac{\|a_{t}^i\|_1 b}{r_t^i(h_t^i)}\leq B,
\end{eqnarray}
where $\frac{\|a_{t}^i\|_1 \cdot b}{r_t^i(h_t^i)}$ is the bandwidth needed for transmitting the amount $\|a_{t}^i\|_1 \cdot b$ of bits given the data
rate $r_t^i(h_t^i)$. We can write the constraint compactly as $f(s_t^0, a_t^1, \ldots, a_t^N)\leq 0$.

\subsection{Extensions}
Numerous straightforward extensions of our model are possible. Below we discuss several.

First, note that the video traffic in our model can come from video sequences that are either encoded in real-time or offline and stored in the
memory before the transmission.

Second, although we consider the cellular network in this paper, the model can be readily used for multi-user wireless video transmission over
802.11a/e wireless LANs \cite{vanderSchaarAndreopoulosHu}.

Finally, for clarity of the presentation, we introduce the traffic state $(C_t,x_t)$ and the channel state $h_t$ separately as independent states.
The proposed solution can be applied to a more general multi-user resource allocation and scheduling problem, where the traffic state and the channel
state are correlated. Such a model can be easily adopted in our framework and is useful for real-time video encoding scenarios, where the video
traffic is generated by the video encoder at run-time, while considering the time-varying channel condition and the resource price.

\section{The Design Problem} \label{sec:Problem}
Each user performs packet scheduling based on its state $s_t$. Hence, each user $i$'s strategy can be defined as a mapping $\pi_i(s_t^i) \in
A^i(s_t^i)$, where $A^i(s_t^i)$ is the set of actions available under state $s_t^i$. We allow the set of available actions to depend on the state, in
order to capture the minimum video quality guarantee. For example, we may have a minimum distortion impact reduction requirement $D^i$ for user $i$
at any time, which gives us
$$
A^i(s_t^i) = \left\{a_t^i: \sum_{DU \in C_t^i} q_{DU} y_{t,DU}^i \geq D^i \right\}.
$$

The users aim to maximize their expected long-term payoff. Given its initial state $s_0^i$, each user $i$'s strategy $\pi^i$ induce a probability
distribution over the sequences of states $s_1^i,s_2^i,\ldots$, and hence a probability distribution over the sequences of instantaneous payoffs
$u_0^i, u_1^i,\ldots$. Taking expectation with respect to the sequences of payoffs, we have user $i$'s \emph{long-term payoff} given the initial
state as
\begin{eqnarray}\label{eqn:LongTermAggregatorCost}
U^i(\pi^i|(s_0^i)) = \mathbb{E} \left\{ (1-\delta) \sum_{t=0}^{\infty} \left(\delta^t \cdot u_t^i\right) \right\},
\end{eqnarray}
where $\delta \in [0,1)$ is the discount factor.

The design problem can be formulated as
\begin{eqnarray}\label{eqn:DesignProblem}
& \min_{\bm{\pi}} & \sum_{s_0^1,\ldots,s_0^I} \sum_{i=1}^I
U^i(\pi^i|(s_0^i)) \\
& s.t. & \mathrm{minimum~video~quality~guarantee:} \nonumber \\
&      & \pi^i(s^i) \in A^i(s^i), ~\forall i,s^i, \nonumber \\
&      & \mathrm{resource~constraint:} \nonumber \\
&      & f(s^0,\pi^1(s^1),\ldots,\pi^I(s^I)) \leq 0, ~\forall s^0. \nonumber
\end{eqnarray}
Note that the design problem \eqref{eqn:DesignProblem} is a weakly-coupled MU-MDP as defined by \cite{Hawkins2003}. It is a MU-MDP because there are
multiple users making foresighted decisions. The MU-MDP is weakly-coupled, because the users influence each other only through the resource
constraints, but the users' actions do not directly affect the others' payoffs. Such a weakly-coupled MU-MDP has been studied in a general setting
\cite{Hawkins2003} and in wireless video transmission \cite{FuVDS_JSAC2010}, both adopting a dual decomposition approach based on uniform price (i.e.
the same Lagrangian multiplier for the resource constraints under all the states).

Note also that we sum up the network utility $\sum_{i=1}^I U^i(\pi^i|(s_0^i))$ under all the initial states $(s_0^1,\ldots,s_0^I)$. This can be
interpreted as the expected network utility when the initial state is uniformly distributed. The optimal stationary strategy profile that maximizes
this expected social welfare will also maximize the social welfare given any initial state.

However, the design problem \eqref{eqn:DesignProblem} is very challenging, and has never been solved optimally. To better understand this, let us
assume that a central controller would exist which knows the complete information of the system (i.e. the states, the state transitions, the payoff
functions) at each time step. Then, this central controller can solve the above problem \eqref{eqn:DesignProblem} as a centralized MDP (e.g. using
well-known Value Iteration or Policy Iteration methods) and obtain the solution to the design problem $\bm{\pi}^\star$ and the optimal value function
$U^\star$. However, the multi-user wireless video system we discussed is inherently informationally-decentralized and there is no entity in the
network that possesses the complete information. Moreover, the computational complexity of solving \eqref{eqn:DesignProblem} is prohibitively high.
Hence, our goal is to develop an optimal decentralized algorithm that converges to the optimal solution.

\section{Optimal Foresighted Video Transmission} \label{sec:DesignFramework}
In this section, we show how to determine the optimal foresighted video transmission policies. We propose an algorithm that allows each entity to
make decisions based on its local information and limited information exchange between the BS and the users. Specifically, in each time slot, the BS
sends resource prices to each user and the users send their total numbers of packets to transmit to the BS. The BS keeps updating the resource prices
based on the resource usage by the users, and obtains the optimal resource prices based on which the users' optimal individual packet scheduling
decisions achieve the optimal network utility. An overview of the main results and the structure of the proposed solution is given in
Fig.~\ref{fig:IllustrationSolution}.

\begin{figure}
\centering
\includegraphics[width =3.5in]{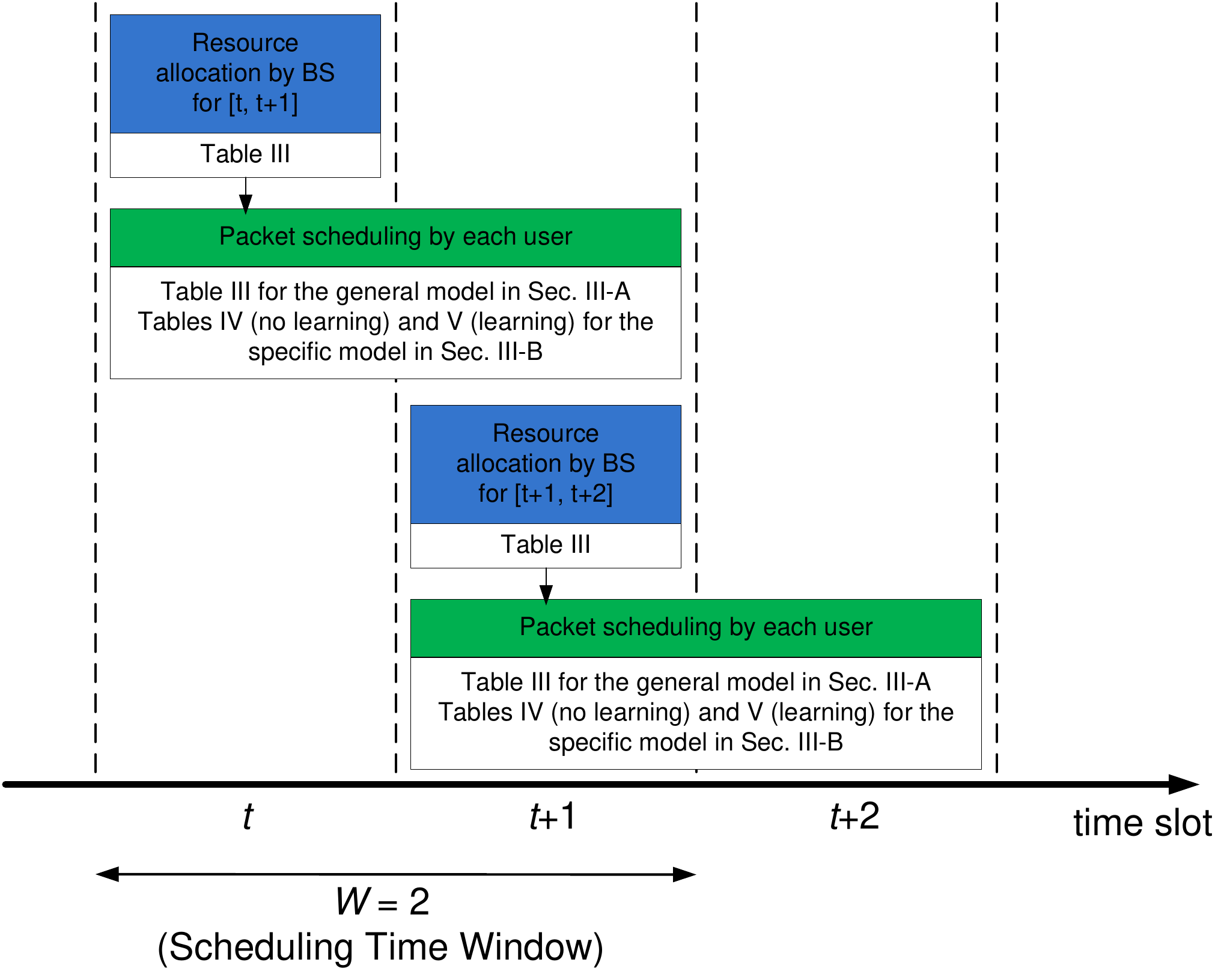}
\caption{Illustration of the resource allocation and packet scheduling in the proposed solution.} \label{fig:IllustrationSolution}
\end{figure}

\subsection{Decoupling of The Users' Decision Problems}
Each user aims to maximize its own long-term payoff $U^i(\pi^i|(s_0^i))$ subject to the constraints. Specifically, given the other users' strategies
$\bm{\pi}^{-i}=(\pi^1,\ldots,\pi^{i-1},\pi^{i+1},\ldots,\pi^I)$ and states $\bm{s}^{-i}=(s^1,\ldots,s^{i-1},s^{i+1},\ldots,s^I)$, each user $i$ solve
the following long-term payoff maximization problem:
\begin{eqnarray}
\pi^i = & \displaystyle\arg\max_{\tilde{\pi}^i} & U_i(\tilde{\pi}^i|(s_0^i)) \\
        & s.t.                     & \tilde{\pi}^i(s^i) \in A^i(s^i), ~\forall s^i, \nonumber \\
        &                          & f(s^0,\tilde{\pi}^i(s^i),\bm{\pi}^{-i}(\bm{s}^{-i})) \leq 0. \nonumber
\end{eqnarray}
Assuming that the user knows all the information (i.e. the other users' strategies $\bm{\pi}^{-i}$ and states $\bm{s}^{-i}$), user $i$'s optimal
value function should satisfy the following:
\begin{eqnarray}\label{eqn:DecisionProblem_EachUser_Original}
V(s^i) = \!\!\!\!\!\!\!\!& \displaystyle\max_{a^i \in A^i(s^i)} & (1-\delta) u^i(s^i,a^i) + \delta \sum_{s^{i \prime}} p^i(s^{i \prime}|s^i,a^i)
V(s^{i \prime}) \nonumber \\
& s.t. & f(s^0,a^i,\bm{\pi}^{-i}(\bm{s}^{-i})) \leq 0.
\end{eqnarray}
Note that the above equations would be the Bellman equations, if the user knew the other users' strategies $\bm{\pi}^{-i}$ and states $\bm{s}^{-i}$
and the BS's state $s^0$ (i.e. the channel states of all the users). However, such information is never known to a particular user. Without such
information, one user cannot solve the decision problem above because the resource constraint contains unknown variables. Hence, we need to separate
the influence of the other users' decisions from each user's decision problem.

One way to decouple the interaction among the users is to remove the resource constraint and add it as a penalty to the objective function. Denote
the Lagrangian multiplier (i.e. the ``price'') associated with the constraint under state $s^0$ as $\lambda^0(s^0)$. Then the penalty at state $s^0$
is
$$
- \lambda^0(s^0) \cdot f(s^0, a^1, \ldots, a^I) = - \lambda^0(s^0) \cdot \left( \sum_{i=1}^I \frac{\|a^i\|_1 b}{r^i(h^i)} - B \right).
$$
Since the term $- \lambda^0(s^0) \cdot \left( \sum_{j \neq i} \frac{\|a^j\|_1 b}{r^j(h^j)} - B \right)$ is a constant for user $i$, we only need to
add $- \lambda^0(s^0) \cdot \frac{\|a^i\|_1 b}{r^i(h^i)}$ to each user $i$'s objective function. Hence, we can rewrite user $i$'s decision problem as
\begin{eqnarray}\label{eqn:DecisionProblem_EachUser}
& & \tilde{V}^{\lambda^i(s^0)}(s^i) \\
&=& \max_{a^i \in A^i(s^i)} (1-\delta) \left[ u_i(s^i,a^i) - \lambda^i(s^0) \cdot \|a^i\|_1 \right] \nonumber \\
& & ~~~~~~~~+~\delta \cdot \sum_{s_i^\prime} \left[p^i(s^{i \prime}|s^i,a^i) \tilde{V}^{\lambda^i(s^0)}(s^{i \prime})\right], \nonumber
\end{eqnarray}
where $\lambda^i(s^0) = \lambda^0(s^0) \cdot \frac{b}{r^i(h^i)}$. By contrasting \eqref{eqn:DecisionProblem_EachUser} with
\eqref{eqn:DecisionProblem_EachUser_Original}, we can see that given the price $\lambda^i$, each user can make decisions based only on its local
information since the resource constraint is eliminated. Note, importantly, that the above decision problem \eqref{eqn:DecisionProblem_EachUser} for
each user $i$ is different from that in \cite{FuVDS_JSAC2010} with uniform price. This can be seen from the term $\lambda^i(s^0) \cdot \|a^i\|_1$ in
\eqref{eqn:DecisionProblem_EachUser}, where the price $\lambda^i(s^0)$ is user-specific and depends on the state, while the uniform price in
\cite{FuVDS_JSAC2010} is a constant $\lambda$. The decision problem \eqref{eqn:DecisionProblem_EachUser} is also different from the subproblem
resulting from dual decomposition in NUM, because it is a foresighted optimization problem that aims to maximize the long-term payoff. This requires
a different method to calculate the optimal Lagrangian multiplier $\lambda^i(s^0)$ than that in NUM.

\subsection{Optimal Decentralized Video Transmission Strategy}
For the general model described in Section~\ref{sec:model}-A, we propose an algorithm used by the BS to iteratively update the prices and by the
users to update their optimal strategies. The algorithm will converge to the optimal prices and the optimal strategy profile that achieves the
minimum total system payoff $U^\star$. The algorithm is described in Table~\ref{table:DecentralizedStrategy}.

\begin{theorem}\label{theorem:Convergence}
The algorithm in Table~\ref{table:DecentralizedStrategy} converges to the optimal strategy profile, namely
$$
\begin{array}{c} \lim_{k \rightarrow
\infty} \left| \sum_{s_0^1,\ldots,s_0^I} \sum_{i=1}^I U_i(\pi^{i,\lambda^i_{k}}|s_0^i) - U^\star \right| = 0 \end{array}.
$$
\end{theorem}
\begin{proof}
See the appendix.
\end{proof}

\begin{table}
\renewcommand{\arraystretch}{1.3}
\caption{Distributed algorithm to compute the optimal decentralized video transmission strategy.} \label{table:DecentralizedStrategy} \centering
\begin{tabular}{l}
\hline
\textbf{Input:} Performance loss tolerance $\epsilon$ \\
\hline
\textbf{Initialization:} Set $k=0$, $\bar{a}^i_0=0,\forall i$, $\lambda^i_0=0,\forall i$. \\
Each user $i$ observes $s^i$, the BS observes $s^0$ \\
\textbf{repeat} \\
~~~~Each user $i$ solves the packet scheduling problem \eqref{eqn:DecisionProblem_EachUser} to obtain $\pi^{i,\lambda_k^{i}(s^0)}$ \\
~~~~Each user $i$ submits its bandwidth request $\|\pi^{i,\lambda_k^{i}(s^0)}(s^i)\|_1 \cdot \frac{b}{r^i(h^i)}$ \\
~~~~The BS updates $\bar{a}^i_{k+1}=\bar{a}^i_{k} + \|\pi^{i,\lambda_k^{i}(s^0)}(s^i)\|_1 \cdot \frac{b}{r^i(h^i)}$ for all $i$ \\
~~~~The BS updates the prices (stochastic subgradient update): \\
~~~~~~~~$\lambda^i_{k+1}(s^0) = \lambda^i_{k}(s^0) + \frac{1}{k+1} \cdot \left( \sum_{i=1}^I \frac{\|\pi^{i,\lambda_k^{i}(s^0)}(s^i)\|_1 \cdot b}{r^i(h^i)} - B \right)$\\
\textbf{until} $\|\lambda^i_{k+1}(s^0)-\lambda^i_{k}(s^0)\| \leq \epsilon$ \\
\hline
\textbf{Output:} optimal prices $\{\lambda^i(s^0)\}_{i=1}^I$, optimal schedules $\{\pi^{i,\lambda^i(s^0)}(s^i)\}_{i=1}^I$ \\
\hline
\end{tabular}
\end{table}

\begin{figure}
\centering
\includegraphics[width =3.5in]{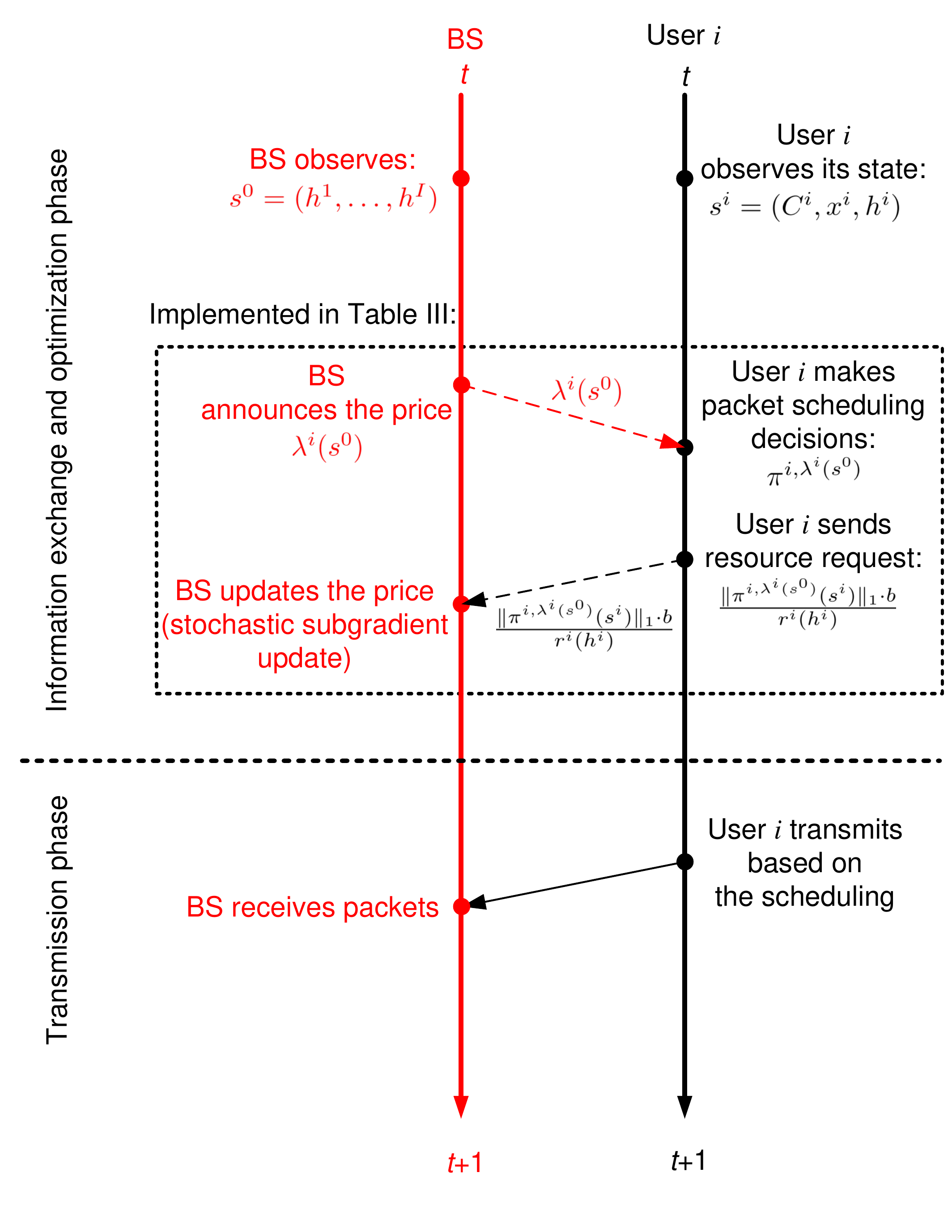}
\caption{Illustration of the interaction between the BS and user $i$ (i.e. their decision making and information exchange) in one period.}
\label{fig:TimeLine}
\end{figure}

We illustrate the the BS's and users' updates and their information exchange in one time slot in Fig.~\ref{fig:TimeLine}. The beginning of one time
slot is the information exchange and optimization phase, in which the BS and the users exchange information such that the BS can update the resource
prices and the users can make optimal packet scheduling decisions. Specifically, given the resource price $\lambda_k^{i}(s^0)$, each user $i$ makes
individual packet scheduling decisions $\pi^{i,\lambda_k^{i}(s^0)}$, and sends the BS its bandwidth request $\|\pi^{i,\lambda_k^{i}(s^0)}(s^i)\|_1
\cdot \frac{b}{r^i(h^i)}$. Then the BS updates the prices based on the users' bandwidth requests using the stochastic subgradient method, which can
be performed easily. The difference from the dual decomposition in NUM is that each user's decision problem in our work is a foresighted optimization
problem aiming to maximize the long-term, instead of instantaneous, payoff. Our algorithm is also different from the algorithm in
\cite{FuVDS_JSAC2010} in that we have different prices for different users.

From Fig.~\ref{fig:TimeLine}, we can clearly see what information (namely resource prices $\lambda_k^{i}(s^0)$ and bandwidth requests
$\|\pi^{i,\lambda_k^{i}(s^0)}(s^i)\|_1 \cdot \frac{b}{r^i(h^i)}$) is exchanged. The amount of information exchange is small ($O(I)$), compared to the
amount of information required by each user to solve the decision problem \eqref{eqn:DecisionProblem_EachUser_Original} directly ($\prod_{j\neq i}
|S_i|$ states plus the strategies $\bm{\pi}_{-i}$). In other words, the algorithm enables the entities to exchange a small amount ($O(I)$) of
information and reach the optimal video transmission strategy that achieves the same performance as when each entity knows the complete information
(i.e. the states and the strategies of all the entities) about the system.

\subsection{Optimal Packet Scheduling}
In the previous subsection, we propose an algorithm of optimal foresighted resource allocation and packet scheduling for the general video
transmission model described in Section~\ref{sec:model}-A. In the algorithm, each user's packet scheduling decision is obtained by solving the
Bellman equation \eqref{eqn:DecisionProblem_EachUser} (see Table~\ref{table:DecentralizedStrategy}). The Bellman equation
\eqref{eqn:DecisionProblem_EachUser} can be solved by a variety of standard techniques such as value iteration. However, the computational complexity
of directly applying value iteration may be high, because each user's state contains the information of all DUs and thus each user's state space can
be very large. In the following, we show that for the specific model described in Section~\ref{sec:model}-B, we can greatly simplify the packet
scheduling decision problem. The key simplification comes from the decomposition of each user's packet scheduling problem with multiple DUs into
multiple packet scheduling problems with single DU. In this way, we can greatly reduce the number of states in each single-DU packet scheduling
problem, such that the total complexity of packet scheduling grows linearly, instead of exponentially without decomposition, with the number of DUs.

The decomposition closely follows the decomposition of multiple-DU packet scheduling problems proposed in \cite{FuVDS_CSVT2012}. The only difference
is that the decision problem \eqref{eqn:DecisionProblem_EachUser} in our work has an additional term $\lambda^i(s^0) \cdot \|a^i\|_1$ due to the
price, while such a term does not exist in \cite{FuVDS_CSVT2012} because the \emph{single}-user packet scheduling problem is considered in
\cite{FuVDS_CSVT2012}.

First, we can greatly simplify the multiple-DU packet scheduling problem by the following structural results.

\begin{lemma}[Structural Result]
Suppose $DU_1 \in C_t$ and $DU_2 \in C_t$. If $DU_2$ depends on $DU_1$, we should schedule the packets of $DU_1$ before scheduling the packets of
$DU_2$.
\end{lemma}
\begin{IEEEproof}
The proof is straightforward and similar to the proof of \cite[Lemma~1]{FuVDS_CSVT2012}. If $DU_2$ depends on $DU_1$, then $DU_1$ has higher
distortion impact and earlier deadline, which means that it always achieves a higher distortion reduction to schedule packets of $DU_1$. In addition,
the contributions in energy consumption and resource payment do not depend on which DU the packets come from. Hence, we should always schedule the
packets of $DU_1$ before scheduling the packets of $DU_2$.
\end{IEEEproof}


Although Lemma~1 is straightforward, it greatly simplifies the scheduling problem because we can now take advantage of the partial ordering of the
DUs. However, this still does not solve the scheduling decision for the DUs that are not dependent on each other. Next, we provides the algorithm of
optimal packet scheduling in Table~\ref{table:OptimalPacketScheduling}. The algorithm decomposes the multiple-DU packet scheduling problem into a
sequence of single-DU packet scheduling problems, and determines how many packets to transmit for each DU sequentially. This greatly reduces the
total computational complexity (which is linear in the number of DUs) compared to solving the multiple-DU packet scheduling problem directly (in
which the number of states grows exponentially withe number of DUs). The algorithm is similar to \cite[Algorithm~2]{FuVDS_CSVT2012}. The only
difference is the term $\lambda^i(s^0) \cdot \|a^i\|_1$.

\begin{table}
\renewcommand{\arraystretch}{1.3}
\caption{Optimal packet scheduling algorithm for each user.} \label{table:OptimalPacketScheduling} \centering
\begin{tabular}{l}
\hline
\textbf{Input:} Directed acyclic graph given the current context: $DAG(C_t)$ \\
\hline
\textbf{Initialization:} Set $DAG_1=DAG(C_t)$. \\
\textbf{For} $k=1,\ldots,|C_t|$ \\
~~~~$DU_k = \displaystyle\arg\max_{DU \in root(DAG_k)} \max_{0\leq y \leq x_{t,DU}} (1-\delta) \left[ q_{DU} \cdot y - \lambda^{i}_k(s^0) \cdot y \right]$ \\
~~~~~~~~~~~~~~~~~~~~~~~~~~~~~~~~~~~~~~~~~$+~\delta \cdot \sum_{s^{i \prime}} \left[p^i(s^{i \prime}|s^i,a_{f,t}) \tilde{V}^{i,\lambda^{i,(k)}(s^0)}(s^{i \prime} )\right]$ \\
~~~~$y^*_{t,DU_k} = \displaystyle\arg\max_{0\leq y \leq x_{t,DU_k}} (1-\delta) \left[ q_{DU_k} \cdot y - \lambda^{i}_k(s^0) \cdot y \right]$ \\
~~~~~~~~~~~~~~~~~~~~~~~~~~~~~~~~~~~~~~~~~$+~\delta \cdot \sum_{s^{i \prime}} \left[p^i(s^{i \prime}|s^i,a_{f,t}) \tilde{V}^{i,\lambda^{i,(k)}(s^0)}(s^{i \prime} )\right]$ \\
~~~~$DAG_{k+1} = DAG_k \setminus \{DU_k\}$ \\
\textbf{End For} \\
\hline
\end{tabular}
\end{table}

\subsection{Learning Unknown Dynamics}
In practice, each entity may not know the dynamics of its own states (i.e., its own state transition probabilities) or even the set of its own
states. When the state dynamics are not known a priori, each entity cannot solve their decision problems using the distributed algorithm in
Table~\ref{table:DecentralizedStrategy}. In this case, we can adapt the online learning algorithm based on post-decision state (PDS) in
\cite{FuVDS_CSVT2012}, which was originally proposed for single-user wireless video transmission, to the considered deployment scenario.

The main idea of the PDS-based online learning is to learn the post-decision value function, instead of the value function. Each user $i$'s
post-decision value function is defined as $\tilde{U}^i(\tilde{x}^i,\tilde{h}^i)$, where $(\tilde{x}^i,\tilde{h}^i)$ is the post-decision state. The
difference from the normal state is that the PDS $(\tilde{x}^i,\tilde{h}^i)$ describes the status of the system \emph{after} the scheduling action is
made but before the DUs in the next period arrive. Hence, the relationship between the PDS and the normal state is
$$
\tilde{x}^i = x^i - a^i,~\tilde{h}^i=h^i.
$$
Then the post-decision value function can be expressed in terms of the value function as follows:
\begin{eqnarray*}
\tilde{U}^i(\tilde{x}^i,\tilde{h}^i) = \sum_{x^{i\prime},h^{i \prime}} p^i(x^{i \prime},h^{i \prime}|\tilde{x}^i + a^i,\tilde{h}^i)\cdot
\tilde{V}^i(x^{i \prime}, \tilde{h}^i).
\end{eqnarray*}
In PDS-based online learning, the normal value function and the post-decision value function are updated in the following way:
\begin{eqnarray*}
V^{i}_{k+1}(x^{i}_k, h^{i}_k) &=& \max_{a^i}~(1-\delta) \cdot u^i(x^{i}_k,h^{i}_k,a^i) \\
                              & & ~~+~\delta \cdot U^{i}_k(x^{i}_k+(a^i-l^{i}_k),h^{i}_k), \\
U^{i}_{k+1}(x^{i}_k, h^{i}_k) &=& (1-\frac{1}{k}) U^{i}_k(x^{i}_k, h^{i}_k) \\
                              &+& \frac{1}{k} \cdot V^{i}_k(x^{i}_k-(a^i-l^{i}_k), h^{i}_k).
\end{eqnarray*}
We can see that the above updates do not require any knowledge about the state dynamics. In particular, we propose the decomposed optimal packet
scheduling with PDS-based learning in Table~\ref{table:OptimalPacketSchedulingLearning}. Note that the difference between the learning algorithm in
Table~\ref{table:OptimalPacketSchedulingLearning} with the algorithm assuming statistic knowledge in Table~\ref{table:OptimalPacketScheduling} is
that we use the post-decision state value function instead of the normal value function. It is proved in \cite{FuVDS_CSVT2012} that the PDS-based
online learning will converge to the optimal value function. Hence, the distributed packet scheduling and resource allocation solution in
Table~\ref{table:DecentralizedStrategy} can be modified by letting each user perform the packet scheduling using the PDS-based learning in
Table~\ref{table:OptimalPacketSchedulingLearning}.

\begin{table} \scriptsize
\renewcommand{\arraystretch}{1.3}
\caption{Optimal decomposed packet scheduling algorithm with PDS-based learning.} \label{table:OptimalPacketSchedulingLearning} \centering
\begin{tabular}{l}
\hline
\textbf{Input:} Directed acyclic graph given the current context: $DAG(C_t)$ \\
\hline
\textbf{Initialization:} Set $DAG^1=DAG(C_t)$. \\
\textbf{For} $k=1,\ldots,|C_t|$ \\
~~~~$DU_k = \displaystyle\arg\max_{DU \in root(DAG_k)} \max_{0\leq y \leq x_{t,DU}} (1-\delta) \left[ q_{DU} \cdot y - \lambda^{i}_k(s^0) \cdot y \right]$ \\
~~~~~~~~~~~~~~~~~~~~~~~~~~~~~~~~~~~~~~~~~$+~\delta \cdot U_{DU}(C_t,x_{t,DU}-y,h_t)$ \\
~~~~$y^*_{t,DU_k} = \arg \max_{0\leq y \leq x_{t,DU}} (1-\delta) \left[ q_{DU} \cdot y - \lambda^{i}_k(s^0) \cdot y \right]$ \\
~~~~~~~~~~~~~~~~~~~~~~~~~~~~~~~~~~~~~~~~~$+~\delta \cdot U_{DU}(C_t,x_{t,DU}-y,h_t)$ \\
~~~~$DAG_{k+1} = DAG_k \setminus \{DU_k\}$ \\
\textbf{End For} \\
\hline
\end{tabular}
\end{table}

\begin{table*}
\renewcommand{\arraystretch}{1.3}
\caption{Relationship between the proposed and existing theoretical frameworks.} \label{table:DetailedComparison_Frameworks} \centering
\begin{tabular}{|c|c|c|}
\hline
Framework & Relationship & Representative works \\
\hline
Myopic & $\delta=0$ & \cite{vanderSchaarAndreopoulosHu}--\cite{JiHuangChiang} \\
\hline
Lyapunov optimization & User $i$'s post-decision value function $U^i(s^i)=u^i(s^i,a^i)+\|x^i-a^i+l^i\|_1^2-\|x^i\|_1^2$ & \cite{Neely} \\
\hline
MU-MDP & Lagrangian multiplier $\bm{\lambda}(s^0)=\bm{\lambda}$ for all $s^0$, and $\lambda^i=\lambda^j$ for all users $i,j$ & \cite{FuVDS_JSAC2010} \\
\hline
\end{tabular}
\end{table*}

\subsection{Detailed Comparisons with Existing Frameworks}
Since we have introduced our proposed framework, we can provide a detailed comparison with the existing theoretical framework. The comparison is
summarized in Table~\ref{table:DetailedComparison_Frameworks}.

First, the proposed framework reduces to the myopic optimization framework (repeated NUM) when we set the discount factor $\delta=0$. In this case,
the proposed solution reduces to the myopic solution.

Second, the Lyapunov optimization framework is closely related to the PDS-based online learning. In fact, it could be considered as a special case of
the PDS-based online learning when we set the post-decision value function as $U^i(s^i)=u^i(s^i,a^i)+\|x^i-a^i+l^i\|_1^2-\|x^i\|_1^2$, and choose the
action that maximizes the post-decision value function at run-time. However, the Lyapunov drift in the above post-decision value function depends
only on the total number of packets in the queue, but not on the delay deadlines, the dependency among packets, and the channel condition. In
contrast, in our PDS-based online learning, we explicitly consider the impact of the video traffic (i.e. delay deadlines and dependency) and the
channel condition when updating the normal and post-decision value functions.

Finally, the key difference between our proposed framework and the framework for MU-MDP \cite{FuVDS_JSAC2010} is how we penalize the constraints
$\bm{f}(s^0,\bm{a})$. In particular, the framework in \cite{FuVDS_JSAC2010}, if directly applied in our model, would define only one Lagrangian
multiplier for all the constraints under different states $s^0$. This in general leads to performance loss \cite{FuVDS_JSAC2010}. In contrast, we
define different Lagrangian multipliers to penalize the constraints under different states $s^0$, and enable the proposed framework to achieve the
optimality (which is indeed the case as have been proved in Theorem~\ref{theorem:Convergence}).

\section{Simulation Results}\label{sec:Simulation}
We consider a wireless network with multiple users streaming one of the following three video sequences, ``Foreman'' (CIF resolution, 30 Hz),
``Coastguard'' (CIF resolution, 30 Hz), and ``Mobile'' (CIF resolution, 30 Hz). We use the following system parameters by default and will explicitly
specify when changes are made. The video sequences Foreman and Coastguard are encoded at a bit rate of 512 kb/s, and Mobile is encoded at 1024 kb/s
due to its higher and more complicated video characteristics. Each GOP contains 16 frames and each encoded video frame can tolerate a delay of 266 ms
(namely the duration of 8 frames, or half the duration of a GOP). The length of one time slot is 10 ms, and the scheduling time window $W$ is 266~ms.
The energy consumption function is set as $\rho(h,\|a\|_1) = \sigma^2(2^{\|a\|_1}-1)/|h|^2$
\cite{FuVDS_CSVT2012}\cite{FuVDS_JSAC2010}\cite{BertsekasGallager}, where $|h|^2/\sigma^2=1.4$ ($\approx 1.5$~dB). We set the tradeoff parameter of
distortion reduction and energy consumption as $\beta=1$. The discount factor is $\delta=0.95$.

\subsection{Convergence of the PDS-based Learning Algorithm}
Before comparing against the other solutions, we show that the proposed PDS-based learning algorithm converges to the optimal solution (namely the
optimal value function is learned). The optimal solution is obtained by the proposed algorithm in Table~\ref{table:DecentralizedStrategy} assuming
the statistical knowledge of the system dynamics. We consider a scenario with 10 users streaming the video sequence Foreman. We compare long-term
video quality (i.e. Peak Signal-to-Noise Ratio (PSNR)) achieved by the optimal solution and that achieved by the PDS-based learning algorithm. For
illustrative purpose, we show the convergence of the learning algorithm in terms of long-term video quality only for two users in
Fig.~\ref{fig:Learning}.

\begin{figure}
\centering
\includegraphics[width =3.5in]{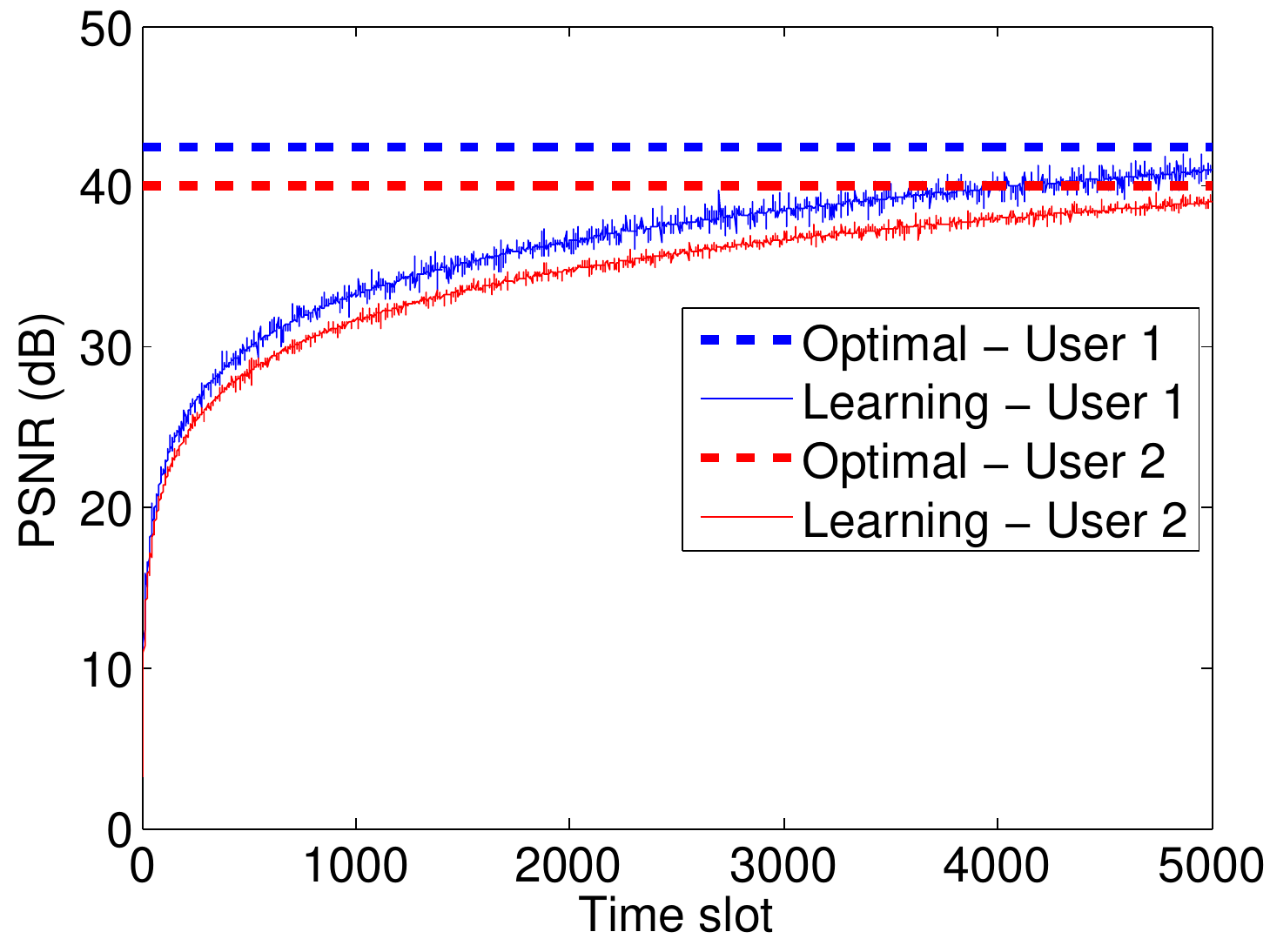}
\caption{Convergence of the PDS-based learning algorithm.} \label{fig:Learning}
\end{figure}

\begin{table}
\renewcommand{\arraystretch}{1.3}
\caption{Resource allocation and packet scheduling used in different solutions.} \label{table:DifferentSolutions} \centering
\begin{tabular}{|c|c|c|}
\hline
 & Resource allocation & Packet scheduling \\
\hline
Myopic \cite{vanderSchaarAndreopoulosHu}--\cite{JiHuangChiang} & myopic & earliest-deadline-first \cite{JiHuangChiang}  \\
\hline
Lyapunov \cite{Neely} & same as proposed (optimal) & based on Lyapunov drift \\
\hline
MU-MDP \cite{FuVDS_JSAC2010} & based on uniform price & same as proposed (optimal) \\
\hline
Proposed & as in Table~\ref{table:DecentralizedStrategy} (optimal) & as in Table~\ref{table:OptimalPacketScheduling} or \ref{table:OptimalPacketSchedulingLearning} (optimal) \\
\hline
\end{tabular}
\end{table}

\subsection{Comparison Against Existing Solutions}
Now we compare against the myopic solution (i.e. repeated NUM) \cite{vanderSchaarAndreopoulosHu}--\cite{JiHuangChiang}, the Lyapunov optimization
solution (adapted from \cite{Neely} for video transmission), and the MU-MDP solution \cite{FuVDS_JSAC2010}. In Table~\ref{table:DifferentSolutions},
we summarize the resource allocation and packet scheduling schemes adopted in the above solutions. Note that the Lyapunov solution is proposed for
the single-user problem without resource allocation. To fairly compare the optimal packet scheduling with the packet scheduling in the Lyapunov
solution that ignores video traffic, we adopt the proposed optimal resource allocation scheme in the Lyapunov solution. Similarly, since the proposed
and MU-MDP solutions use the same optimal packet scheduling, we can fairly compare the proposed resource allocation with the suboptimal foresighted
resource allocation based on uniform price.

\subsubsection{Resource Allocation and Packet Scheduling Decisions Made by Different Solutions}
Before comparing the performance of different solutions, we illustrate the resource allocation and packet scheduling of different solutions under one
realization of a sequence of traffic states and channel states. From this illustration, we can better understand the differences of the resource
allocation and packet scheduling decisions made by different solutions and their impact on the video quality.

\begin{table*} \scriptsize
\renewcommand{\arraystretch}{1.0}
\caption{Resource Allocation and Packet Scheduling of The Myopic (Repeated NUM) Solution.} \label{table:IllustrationMyopic} \centering
\begin{tabular}{|c|c|c|c|c|c|}
\hline \multirow{2}{*}{Traffic states}
                    & I(40), P(10), B(10) & P(10), B(10), I(40) & I(40), P(10), B(10) & P(10), B(10), I(40) & I(40), P(10), B(10) \\
                    & I(40), P(10)        & P(10), P(10)        & P(0), I(40)         & I(20), P(10)        & P(10), P(10) \\
\hline
Channel state       & good                & bad                 & bad                 & bad                 & good \\
\hline
Resource allocation & (0.50, 0.50)        & (0.50, 0.50)        & (0.50, 0.50)        & (0.50, 0.50)        & (0.50, 0.50) \\
\hline \multirow{2}{*}{Packet scheduling}
                    & I(30), P(0), B(0)   & P(10), B(10), I(0)  & I(20), P(0), B(0)   & P(10), B(10), I(0)  & I(30), P(0), B(0) \\
                    & I(30), P(0)         & P(10), P(10)        & P(0), I(20)         & I(20), P(0)         & P(10), P(10) \\
\hline
Packet loss         & I(10), both users   &                     & I(20), user 1       &                     & I(10), user 1 \\
\hline
\end{tabular}
\end{table*}

\begin{table*} \scriptsize
\renewcommand{\arraystretch}{1.0}
\caption{Resource Allocation and Packet Scheduling of The Lyapunov Solution.} \label{table:IllustrationLyapunov} \centering
\begin{tabular}{|c|c|c|c|c|c|}
\hline \multirow{2}{*}{Traffic states}
                    & I(40), P(10), B(10) & P(10), B(10), I(40) & I(32), P(10), B(10) & P(10), B(10), I(40) & I(40), P(10), B(10) \\
                    & I(40), P(10)        & P(10), P(10)        & P(8), I(40)         & I(33), P(10)        & P(8), P(10) \\
\hline
Channel state       & good                & bad                 & bad                 & bad                 & good \\
\hline
Resource allocation & (0.50, 0.50)        & (0.72, 0.28)        & (0.63, 0.37)        & (0.35, 0.65)        & (0.64, 0.36) \\
\hline \multirow{2}{*}{Packet scheduling}
                    & I(30), P(0), B(0)   & P(10), B(10), I(8)  & I(25), P(0), B(0)   & P(10), B(3), I(0)   & I(36), P(0), B(0) \\
                    & I(30), P(0)         & P(10), P(2)         & P(8), I(7)          & I(27), P(0)         & P(8), P(10) \\
\hline
Packet loss         & I(10), both users   &                     & I(7), user 1        & I(6), user 2        & I(4), user 1 \\
\hline
\end{tabular}
\end{table*}

\begin{table*} \scriptsize
\renewcommand{\arraystretch}{1.0}
\caption{Resource Allocation and Packet Scheduling of The MU-MDP Solution.} \label{table:IllustrationMUMDP} \centering
\begin{tabular}{|c|c|c|c|c|c|}
\hline \multirow{2}{*}{Traffic states}
                    & I(40), P(10), B(10) & P(10), B(10), I(40) & I(40), P(10), B(10) & P(0), B(7), I(40) & I(24), P(10), B(10) \\
                    & I(40), P(10)        & P(10), P(10)        & P(6), I(40)         & I(22), P(10)      & P(8), P(10) \\
\hline
Channel state       & good                & bad                 & bad                 & bad               & good \\
\hline Resource allocation & \multirow{2}{*}{(0.50, 0.50)}
                                          & \multirow{2}{*}{(0.46, 0.22)}
                                                                & \multirow{2}{*}{(0.45, 0.42)}
                                                                                      & \multirow{2}{*}{(0.45, 0.35)}
                                                                                                          & \multirow{2}{*}{(0.54, 0.31)} \\
(before scaling)    & & & & & \\
\hline \multirow{2}{*}{Packet scheduling}
                    & I(30), P(0), B(0)   & P(0), B(0), I(26)   & I(12), P(10), B(0)  & P(0), B(7), I(16) & I(24), P(10), B(10) \\
                    & I(30), P(0)         & P(10), P(4)         & P(0), I(18)         & I(15), P(2)       & P(8), P(10) \\
\hline
Packet loss         & I(10), both users   &                     & P(6), user 2        & I(7), user 2      & \\
\hline
\end{tabular}
\end{table*}

\begin{table*} \scriptsize
\renewcommand{\arraystretch}{1.0}
\caption{Resource Allocation and Packet Scheduling of The Proposed Solution.} \label{table:IllustrationProposed} \centering
\begin{tabular}{|c|c|c|c|c|c|}
\hline \multirow{2}{*}{Traffic states}
                    & I(40), P(10), B(10) & P(10), B(10), I(40) & I(40), P(10), B(10) & P(0), B(7), I(40) & I(27), P(10), B(10) \\
                    & I(40), P(10)        & P(10), P(10)        & P(8), I(40)         & I(25), P(10)      & P(8), P(10) \\
\hline
Channel state       & good                & bad                 & bad                 & bad               & good \\
\hline
Resource allocation & (0.50, 0.50)        & (0.72, 0.28)        & (0.63, 0.37)        & (0.35, 0.65)      & (0.64, 0.36) \\
\hline \multirow{2}{*}{Packet scheduling}
                    & I(30), P(0), B(0)   & P(0), B(0), I(28)   & I(12), P(10), B(3)  & P(0), B(0), I(13) & I(27), P(9), B(0) \\
                    & I(30), P(0)         & P(10), P(2)         & P(0), I(15)         & I(25), P(2)       & P(8), P(10) \\
\hline
Packet loss         & I(10), both users   &                     & P(8), user 2        & B(7), user 2      & \\
\hline
\end{tabular}
\end{table*}

For illustrative purposes, we consider a simple but representative system with two users streaming video sequence Foreman. User 1 encodes the
sequence with GOP structure ``IPB'' and user 2 with GOP structure ``IPP''. We fix the numbers of packets in the I, P, B frames to be 40, 10, 10,
respectively. The scheduling time window is 2. The channel state has two values, ``good'' and ``bad''. When the channel state is good (bad), the two
users can transmit up to 60 (40) packets in total. We illustrate the resource allocation and packet scheduling of different solutions in
Tables~\ref{table:IllustrationMyopic}--\ref{table:IllustrationProposed}.

In the myopic solution, the resource allocation is static and determined based on the total distortion impact of the users. Since the users have very
similar GOP structures, the resource allocation is fixed at (0.50, 0.50) in all the time slots. The packet scheduling is EDF. From
Table~\ref{table:IllustrationMyopic}, we can see that severe packet loss occurs due to three consecutive bad channel states. In particular, in the
second bad channel state, user 1's I frame approaches its transmission deadline and needs to be transmitted, while user 2's I frame can wait to be
transmitted in the next time slot. However, due to the static resource allocation, it can only transmit 20 packets and has to discard 20 packets of
its I frame. The same problem occurs in the Lyapunov solution. However, since we adopt the proposed optimal resource allocation, only 7 packets of
the I frame are lost in the second bad channel state. Another problem of both myopic and Lyapunov solutions is that they schedule based on deadlines
without considering the distortion impact of the packets. For example, in the third bad channel state, user 1 has context (P, B, I) with the I frame
of the next GOP having the latest deadline. As a consequence, both solutions schedule the packets of P and B frames, leaving many packets of I frame
in the buffer. This results in the loss of 10 packets and 4 packets of the I frame in the myopic and Lyapunov solutions, respectively.

We can see from Table~\ref{table:IllustrationProposed} that the proposed solution achieves much better performance due to foresighted resource
allocation and packet scheduling. First, notice that as in all the other solutions, the I frame experiences packet losses in the first time slot,
where there are 80 packets of I frame in total and the transmission capacity is 60 packets. This packet loss is inevitable since there is no room for
foresighted resource allocation and packet schedule ahead of the first time slot. The distinction of the proposed solution is evident in the bad
channel states. In the second bad channel state, it allocates much more resource (72\%) to user 1, although both users have a I frame to transmit.
This is because user 1's I frame reaches its transmission deadline at this time slot. Hence, such an allocation meets the urgent need for resource of
user 1 and avoids packet loss in its I frame. In addition, when the context is (P, B, I) with I frame of the next GOP having the latest deadline, the
proposed solution will still schedule some packets of I frames at the expense of losing some packets of P and B frames. Due to foresighted resource
allocation and packet scheduling, the proposed solution loses no packet of the I frames, which will result in much higher video quality than the
other solutions.

The MU-MDP solution illustrated in Table~\ref{table:IllustrationMUMDP} uses the same optimal packet scheduling as proposed and a suboptimal resource
allocation. As we have discussed earlier, the MU-MDP solution uses a uniform price, which needs to be set high enough to avoid the violation of the
resource constraints. Hence, the resource allocation is relatively conservative, as we can see from Table~\ref{table:IllustrationMUMDP}. Although the
initial resource allocation shown in the table will be scaled up by the BS as in \cite{FuVDS_JSAC2010} to ensure full resource utilization, the
proportion of the resource allocation is suboptimal. This suboptimal resource allocation results in the loss of 7 packets of the I frame in the third
bad channel state, because user 2 does not get enough resources.

\subsubsection{Heterogeneity of Video Traffic} We compare the proposed solution with existing solutions in terms of PSNR and energy consumption.
We change the tradeoff parameter $\beta$ from 1 to 30 to get different PSNR and energy consumption tradeoffs. We assume that there are three users
streaming the three different video sequences, respectively. The results are listed in Table~\ref{table:EnergyPSNR}. We can see that the proposed
solution can achieve  for all the users 7~dB PSNR improvement compared to the myopic solution, 5~dB PSNR improvement compared to the Lyapunov
solution, and 3~dB PSNR improvement compared to the MU-MDP solution with uniform price. Moreover, the proposed solution can achieve high PSNR for all
the three different video sequences, while existing solutions may result in low quality (e.g. less than 30~dB PSNR) especially for the more
challenging ``Coastguard'' and ``Mobile'' videos.

\begin{table*}
\renewcommand{\arraystretch}{1.3}
\caption{Comparisons of PSNR under different energy consumptions.} \label{table:EnergyPSNR} \centering
\begin{tabular}{|c|c|c|c|}
\hline
Energy consumption (Joule) & 0.08 & 0.10 & 0.15 \\
\hline
Myopic (repeated NUM) \cite{vanderSchaarAndreopoulosHu}--\cite{JiHuangChiang} & (30, 24, 24)~dB &  (40, 35, 25)~dB & (45, 45, 34)~dB  \\
\hline
Lyapunov \cite{Neely} & (31, 26, 26)~dB & (42, 38, 27)~dB & (47, 47, 37)~dB  \\
\hline
MU-MDP \cite{FuVDS_JSAC2010} & (34, 28, 27)~dB & (44, 40, 28)~dB & (50, 48, 38)~dB  \\
\hline
Proposed & (37, 32, 30)~dB & (46, 43, 32)~dB & (54, 51, 42)~dB  \\
\hline
\end{tabular}
\end{table*}

\subsubsection{Scaling With The Number of Users}
Now we compare the proposed solution with existing solutions in terms of the average PSNR across users when the number of users increases. The energy
consumption is fixed at 0.15 Joule. We assume that all the users stream the ``Foreman'' video sequence. Table~\ref{table:EnergyPSNR_NumberOfUser}
summarizes the results for different numbers of users. We can see that the performance gain of the proposed solution increases with the number of
users. In particular, the proposed solution can achieve high enough PSNR for high-quality video transmission (e.g. larger than 30 dB PSNR) even when
there are 20 users sharing the same resource. On the contrary, the myopic solution, the Lyapunov solution, and the MU-MDP solution may not achieve
high-quality video transmission when the number of users exceeds 9, 9, and 13, respectively. 

\begin{table*}
\renewcommand{\arraystretch}{1.3}
\caption{Comparisons of PSNR under different numbers of users.} \label{table:EnergyPSNR_NumberOfUser} \centering
\begin{tabular}{|c|c|c|c|c|c|c|c|c|}
\hline
Number of users & 3 & 5 & 7 & 9 & 11 & 13 & 15 & 20 \\
\hline
Myopic (repeated NUM) \cite{vanderSchaarAndreopoulosHu}--\cite{JiHuangChiang} & 45~dB & 41~dB & 36~dB & 31~dB & 27~dB & 24~dB & 20~dB & 15~dB  \\
\hline
Lyapunov \cite{Neely} & 47~dB & 45~dB & 42~dB & 36~dB & 29~dB & 26~dB & 23~dB & 20~dB  \\
\hline
MU-MDP \cite{FuVDS_JSAC2010} & 50~dB & 46~dB & 44~dB & 39~dB & 32~dB & 30~dB & 26~dB & 23~dB  \\
\hline
Proposed & 54~dB & 52~dB & 49~dB & 44~dB & 38~dB & 36~dB & 32~dB & 30~dB  \\
\hline
\end{tabular}
\end{table*}

\subsubsection{Different Channel Conditions and Encoding Rates}
We compare different solutions under different channel conditions and encoding rates. We consider a scenario with 10 users streaming the video
sequence Foreman. The channel condition (i.e. $|h|^2/\sigma^2$) varies from 1.0 dB to 2.0 dB, and the encoding rate varies from 256 kb/s to 1024
kb/s. We summarize the PSNR achieved by different solutions in Table~\ref{table:ChannelEncodingRate}. We can see that the proposed solution
outperforms the other solutions in various channel conditions and encoding rates considered, with an average improvement in PSNR of 10~dB over the
myopic solution, 7~dB over the Lyapunov solution, and 5~dB over the MU-MDP solution. We can also observe the importance of selecting the appropriate
bitstream bit-rate given the channel condition. When the channel is bad (e.g. 1.0 dB), we may want to switch to a lower bit-rate bitstream, while we
may prefer a higher bit-rate when the channel is good (e.g. 2.0 dB).

\begin{table*}\scriptsize
\renewcommand{\arraystretch}{1.0}
\caption{Comparisons of PSNR under different channel conditions and encoding rates.} \label{table:ChannelEncodingRate} \centering
\begin{tabular}{|c|c|c|c|c|c|c|c|c|c|}
\hline
(dB, kb/s) & (1.0, 256) & (1.0, 512) & (1.0, 1024) & (1.5, 256) & (1.5, 512) & (1.5, 1024) & (2.0, 256) & (2.0, 512) & (2.0, 1024) \\
\hline
Myopic \cite{vanderSchaarAndreopoulosHu}--\cite{JiHuangChiang} & 27~dB & 26~dB & 23~dB & 29~dB & 29~dB & 30~dB & 30~dB & 32~dB & 34~dB \\
\hline
Lyapunov \cite{Neely} & 31~dB & 30~dB & 28~dB & 32~dB & 32~dB & 33~dB & 34~dB & 34~dB & 35~dB \\
\hline
MU-MDP \cite{FuVDS_JSAC2010} & 34~dB & 33~dB & 31~dB & 35~dB & 36~dB & 36~dB & 38~dB & 39~dB & 40~dB \\
\hline
Proposed & 38~dB & 38~dB & 36~dB & 40~dB & 41~dB & 42~dB & 43~dB & 44~dB & 46~dB \\
\hline
\end{tabular}
\end{table*}

\subsection{Comparison of Resource Allocation Under Different Packet Scheduling Algorithms}
The proposed solution consists of two parts: packet scheduling and resource allocation. The proposed solution is optimal when each user adopts the
optimal packet scheduling. However, the optimal packet scheduling is achieved by solving a MDP, which may have high computational complexity (even
though we have decoupled the users' decision problems and thus have reduced the state space of each user's MDP). In practice, however, the users may
adopt simpler but suboptimal packet scheduling algorithms. Next, we demonstrate that such simpler packet scheduling algorithms can be easily combined
with our proposed resource allocation scheme, and show that when using simpler scheduling algorithms, our resource allocation scheme still
outperforms the resource allocation of the myopic and MU-MDP solutions with uniform price. Specifically, we consider the following three simpler
scheduling algorithms:
\begin{itemize}
\item EDF (Earliest Deadline First) scheduling: schedule all the packets of the frame with the earliest deadline first within its deadline.
\item FIFO (First-In First-Out) scheduling: schedule all the packets of the frame that comes in the buffer first within its deadline.
\item HDF (Highest Distortion First) scheduling: schedule all the packets of the frame with the highest distortion impact first within its deadline.
\end{itemize}

We consider a scenario with 10 users streaming the video sequence Foreman. The channel condition (i.e. $|h|^2/\sigma^2$) is 1.5 dB, and the encoding
rate is 512 kb/s. In Table~\ref{table:Scheduling}, we compare the proposed resource allocation with myopic resource allocation and the resource
allocation in the MU-MDP solution under different scheduling algorithms. We can see that our proposed resource allocation scheme outperforms the
other resource allocation schemes even when we use simpler scheduling algorithms. Moreover, the performance loss induced by using simpler scheduling
algorithms is smaller under our proposed solution (around 3 dB), compared to the myopic solution (around 5 dB).

\begin{table}
\renewcommand{\arraystretch}{1.3}
\caption{Comparisons of PSNR under simpler scheduling algorithms.} \label{table:Scheduling} \centering
\begin{tabular}{|c|c|c|c|c|}
\hline
Packet scheduling & EDF & FIFO & HDF & Optimal \\
\hline
Myopic \cite{vanderSchaarAndreopoulosHu}--\cite{JiHuangChiang} &  24~dB &  22~dB &  25~dB &  29~dB \\
\hline
MU-MDP \cite{FuVDS_JSAC2010} &  34~dB &  32~dB &  33~dB &  36~dB\\
\hline
Proposed &  38~dB &  37~dB &  38~dB &  41~dB\\
\hline
\end{tabular}
\end{table}

\section{Conclusion}\label{sec:Conclusion}
We propose the optimal foresighted resource allocation and packet scheduling for multi-user wireless video transmission. The proposed solution
achieves the optimal long-term video quality subject to each user's minimum video quality guarantee, by dynamically allocating resources among the
users and dynamically scheduling the users' packets while taking into account the dynamics of the video traffic and channel states. We develop a
low-complexity algorithm that can be implemented by the BS and the users in an informationally-decentralized manner and converges to the optimal
solution. Through extensive simulation, we demonstrate the performance gain of our proposed solution over existing solutions under a wide range of
deployment scenarios: different number of users, different channel conditions, different video encoding rates, and different (simpler but suboptimal)
packet scheduling algorithms. The simulation results show that our proposed solution can achieve significant improvements in PSNR of up to 7~dB
compared to myopic solutions and of up to 3~dB compared to state-of-the-art foresighted solutions.

\appendices
\section{Proof of Theorem~\ref{theorem:Convergence}}\label{proof:Convergence}

Due to limited space, we give a detailed proof sketch. The proof consists of three key steps. First, we prove that by penalizing the constraints
$f(s^0,a^1,\ldots,a^I)\leq0$ into the objective function, the decision problems of different entities can be decentralized. Hence, we can derive
optimal decentralized strategies for different entities under given Lagrangian multipliers. Then we prove that the update of Lagrangian multipliers
converges to the optimal ones under which there is no duality gap between the primal problem and the dual problem, due to the convexity assumptions
made on the cost functions. Finally, we validate the calculation of the prices.

First, suppose that there is a central controller that knows everything about the system. Then the optimal strategy to the design problem
\eqref{eqn:DesignProblem} should result in a value function $V^*$ that satisfies the following Bellman equation: for all $s^1,\ldots,s^I$, we have
\begin{eqnarray}\label{eqn:Bellman_Original}
&& V^*(s^1,\ldots,s^I) = \\
& \max_{a^1,\ldots,a^I}& \!\!\!\!(1-\delta) \sum_{i=1}^I u^i(s^i,a^i)  + \delta \cdot \mathbb{E} \left\{ V^*(s^{1\prime},\ldots,s^{I\prime}) \right\} \nonumber \\
& s.t. & \!\!\!\!f(s^0,a^1,\ldots,a^I) \leq 0. \nonumber
\end{eqnarray}

Defining a Lagrangian multiplier $\lambda(s^0)\in\mathbb{R}_+^N$ associated with the constraints $f(s^0,a^1,\ldots,a^I) \leq 0$, and penalizing the
constraints on the objective function, we get the following Bellman equation:
\begin{eqnarray}\label{eqn:Bellman_Lambda}
&& V^{\bm{\lambda}}(s^1,\ldots,s^I) = \\
& \max_{a^1,\ldots,a^I}& \!\!\!\!(1-\delta) \!\! \left[  \sum_{i=1}^I u^i(s^i,a^i) + \lambda(s^0) f(s^0,a^1,\ldots,a^I) \! \right] \nonumber \\
&& + \delta \cdot \mathbb{E} \left\{ V^*(s^{1\prime},\ldots,s^{I\prime}) \right\}. \nonumber
\end{eqnarray}

In the following lemma, we can prove that \eqref{eqn:Bellman_Lambda} can be decomposed.

\begin{lemma}\label{lemma:DecompositionValueFunction}
The optimal value function $V^{\lambda}$ that solves \eqref{eqn:Bellman_Lambda} can be decomposed as $V^{\lambda}(s^1,\ldots,s^I) = \sum_{i=1}^I
V^{\lambda,i}(s_i)$ for all $(s^1,\ldots,s^I)$, where $V^{\lambda,i}$ can be computed by user $i$ locally by solving
\begin{eqnarray}
V^{\lambda,i}(s^i) = \max_{a^i} ~ (1-\delta) \cdot \left[u^i(s^i,a^i)+\lambda f^i(s^0,a^i)\right] \\
+ \delta \cdot \sum_{s^{i\prime}} p^i(s^{i \prime}|s^i,a^i) V^{\lambda,i}(s^{i \prime}). \nonumber
\end{eqnarray}
\end{lemma}
\begin{IEEEproof}
This can be proved by the independence of different entities' states and by the decomposition of the constraints $f(s^0,a^1,\ldots,a^I)$.
Specifically the constraints $f(s^0,a^1,\ldots,a^I)$ are linear with respect to the actions $a_1,\ldots,a_I$. As a result, we can decompose the
constraints as $f(s^0,a^1,\ldots,a^I) = \sum_{i=1}^I f^i(s^0,a^i)$.
\end{IEEEproof}

We have proved that by penalizing the constraint $f(s^0,a^1,\ldots,a^I)$ using Lagrangian multiplier $\lambda(s_0)$, different entities can compute
the optimal value function $V^{\lambda(s^0),i}$ distributively. Due to the convexity assumptions on the cost functions, we can show that the primal
problem \eqref{eqn:DesignProblem} is convex. Hence, there is no duality gap. In other words, at the optimal Lagrangian multipliers $\lambda^*(s^0)$,
the corresponding value function $V^{\lambda^*(s^0)}(s^1,\ldots,s^I) = \sum_{i=1}^I V^{\lambda^*(s^0),i}(s^i)$ is equal to the optimal value function
$V^*(s^1,\ldots,s^I)$ of the primal problem \eqref{eqn:Bellman_Original}. It is left to show that the update of Lagrangian multipliers converge to
the optimal ones. It is a well-known result in dynamic programming that $V^{\lambda(s^0)}$ is convex and piecewise linear in $\lambda(s^0)$, and that
the subgradient is $f(s^0,a^1,\ldots,a^I)$. Note that we use the sample mean of $a^i$, whose expectation is the true mean value of $a^i$. Since
$f(s^0,a^1,\ldots,a^I)$ is linear in $a^i$, the subgradient calculated based on the sample mean has the same mean value as the subgradient calculated
based on the true mean values. In other words, the update is a stochastic subgradient descent method. It is well-known that when the stepsize
$\Delta(k)=\frac{1}{k+1}$, the stochastic subgradient descent will converge to the optimal $\lambda^*$.

Finally, we can write the prices by taking the derivatives of the penalty terms. For user $i$, its penalty is $\lambda(s^0) \cdot f^i(s^0,a^i)$.
Hence, its price is
\begin{eqnarray}
\frac{\partial \lambda^i(s^0) \cdot f^i(s^0,a^i)}{\partial a^i} = \lambda(s_0) \cdot \frac{\partial f^i(s^0,a^i)}{\partial a^i} = \lambda(s_0) \cdot
\frac{b}{r^i(h^i)}. \nonumber
\end{eqnarray}

\end{document}